\documentclass[11pt, a4paper]{article}
\usepackage[a4paper]{anysize}

\usepackage{makeidx}         
\usepackage{graphicx}        
\usepackage{multicol}        
\usepackage[bottom]{footmisc}

\usepackage[round,authoryear]{natbib}

\usepackage{amssymb}
\usepackage{amsmath, amsthm}
\usepackage{amsfonts, graphicx, stmaryrd, fancyhdr}

\usepackage[a4paper]{anysize}
\usepackage[english]{babel}

\usepackage{enumerate}
\usepackage{xcolor}

\usepackage{
            epsfig
           }
\usepackage{subfigure}
\usepackage[round,authoryear]{natbib}

\newcommand{\ID}{\mathbb{D}}
\newcommand{\IE}{\mathbb{E}}

\newcommand{\IL}{\mathbb{L}}
\newcommand{\IN}{\mathbb{N}}

\newcommand{\IP}{\mathbb{P}}
\newcommand{\IQ}{\mathbb{Q}}

\newcommand{\IR}{\mathbb{R}}
\newcommand{\R}{\mathbb{R}}

\newcommand{\cA}{\mathcal{A}}

\newcommand{\cE}{\mathcal{E}}
\newcommand{\cF}{\mathcal{F}}
\newcommand{\cG}{\mathcal{G}}
\newcommand{\cH}{\mathcal{H}}

\newcommand{\cK}{\mathcal{K}}

\newcommand{\cS}{\mathcal{S}}

\newcommand{\ud}{\mathrm{d}}
\newcommand{\uds}{\mathrm{d}s}
\newcommand{\udt}{\mathrm{d}t}
\newcommand{\udu}{\mathrm{d}u}
\newcommand{\udws}{\mathrm{d}W_s}

\newcommand{\1}{\textbf{1}}
\newcommand{\ti}{{t_{i}}}
\newcommand{\tip}{{t_{i+1}}}

\newcommand{\bit}{\begin{itemize}}
\newcommand{\eit}{\end{itemize}}

\theoremstyle{plain}
\newtheorem{theorem}{Theorem}[section]
\newtheorem{lemma}[theorem]{Lemma}

\newtheorem{remark}[theorem]{Remark}
\newtheorem{example}[theorem]{Example}

\newtheorem{corollary}[theorem]{Corollary}

\begin{document}

\title{Results on numerics for FBSDE with drivers of quadratic growth}
\author{\normalsize Peter Imkeller \\
        \small  Dept. of Mathematics  \\
        \small  Humboldt University Berlin \\
        \small  Unter den Linden 6\\
        \small  10099 Berlin \\
        \small  imkeller@math.hu-berlin.de
           \and
        \normalsize Gon\c calo Dos Reis \\
        \small  CMAP \\
        \small  \'Ecole Polytechnique \\
        \small  Route de Saclay \\
        \small  91128 Palaiseau Cedex \\
        \small  dosreis@cmap.polytechnique.fr
           \and
        \normalsize  Jianing Zhang \\
        \small  Dept. of Mathematics  \\
        \small  Humboldt University Berlin \\
        \small  Unter den Linden 6\\
        \small  10099 Berlin \\
        \small zhangj@math.hu-berlin.de
\vspace*{0.8cm}}

\maketitle

\abstract{We consider the problem of numerical approximation for
forward-back\-ward stochastic differential equations with drivers of
quadratic growth (qgFBSDE). To illustrate the significance of
qgFBSDE, we discuss a problem of cross hedging of an insurance
related financial derivative using correlated assets. For the
convergence of numerical approximation schemes for such systems of
stochastic equations, path regularity of the solution processes is
instrumental. We present a method based on the truncation of the
driver, and explicitly exhibit error estimates as functions of the
truncation height. We discuss a reduction method to FBSDE with
globally Lipschitz continuous drivers, by using the Cole-Hopf
exponential transformation. We finally illustrate our numerical
approximation methods by giving simulations for prices and optimal
hedges of simple insurance derivatives.}

\medskip


\medskip



\section{Introduction}

Owing to their central significance in optimization problems for
instance in stochastic finance and insurance, backward stochastic
differential equations (BSDE), one of the most efficient tools of
stochastic control theory, have been receiving much attention in the
last 15 years. A particularly important class, BSDE with drivers of
quadratic growth, for example, arise in the context of utility
optimization problems on incomplete markets with exponential utility
functions, or alternatively in questions related to risk
minimization for the entropic risk measure. BSDE provide the
genuinely stochastic approach of control problems which find their
analytical expression in the Hamilton-Jacobi-Bellman formalism. BSDE
with drivers of this type keep being a source of intensive research.

As Monte-Carlo methods to simulate random processes, numerical
schemes for BSDE provide a robust method for simulating and
approximating solutions of control problems. Much has been done in
recent years to create schemes for BSDE with Lipschitz continuous
drivers (see \cite{04BT} or \cite{phd-elie} and references therein).
The numerical approximation of BSDE with drivers of quadratic growth
(qgBSDE) or systems of forward-backward stochastic equations with
drivers of this kind (qgFBSDE) turned out to be more complicated.
Only recently, in \cite{phd-dosReis}, one of the main obstacles was
overcome. Following \cite{04BT} in the setting of Lipschitz drivers,
the strategy to prove convergence of a numerical approximation
combines two ingredients: regularity of the trajectories of the
control component of a solution pair of the BSDE in the $L^2$-sense,
a tool first investigated in the framework of globally Lipschitz
BSDE by \cite{phd-zhang}, and a convenient a priori estimate for the
solution. See \cite{04BT}, \cite{GLW05}, \cite{delarue-menozzi} or
\cite{07BD} for numerical schemes of BSDE with globally Lipschitz
continuous drivers, and an implementation of these ideas. The main
difficulty treated in \cite{phd-dosReis} consisted of establishing
path regularity for the control component of the solution pair of
the qgBSDE. For this purpose, the control component, known to be
represented by the Malliavin trace of the other component, had to be
thoroughly investigated in a subtle and complex study of Malliavin
derivatives of solutions of BSDE. This study extends a thorough
investigation of smoothness of systems of FBSDE by methods based on
Malliavin's calculus and BMO martingales independently conducted in
\cite{AIdR07} and \cite{07BriCon}. The knowledge of path regularity
obtained this way is implemented in a second step of the approach in
\cite{phd-dosReis}. The quadratic growth part of the driver is
truncated to create a sequence of approximating BSDE with Lipschitz
continuous drivers. Path regularity is exploited to explicitly
capture the convergence rate for the solutions of the truncated BSDE
as a function of the truncation height. The error estimate for the
truncation, which is of high polynomial order, combines with the
ones for the numerical approximation in any existent numerical
scheme for BSDE with Lipschitz continuous drivers, to control the
convergence of a numerical scheme for qgBSDE.  It allows in
particular to establish the convergence order for the approximation
of the control component in the solution process.

An elegant way to avoid the difficulties related to drivers of
quadratic growth, and to fall back into the setting of globally
Lipschitz ones, consists of using a coordinate transform well known
in related PDE theory under the name ``exponential Cole-Hopf
transformation''. The transformation eliminates the quadratic growth
of the driver in the control component at the cost of
producing a transformed driver of a new BSDE which in general lacks
global Lipschitz continuity in the other component. This difficulty
can be avoided by some structure hypotheses on the driver. Once this
is done, the transformed BSDE enjoys global Lipschitz continuity
properties. Therefore the problem of numerical approximation can be
tackled in the framework of transformed coordinates by schemes well
known in the Lipschitz setting. As stated before, this again
requires path regularity results in the $L^2$-sense for the control
component of the solution pair of the transformed BSDE. For
globally Lipschitz continuous drivers \cite{phd-zhang} provides
path regularity under simple and weak additional assumptions such as
$\frac{1}{2}$-H\"older continuity of the driver in the time
variable. The smoothness of the Cole-Hopf transformation allows
passing back to the original coordinates without losing path
regularity. In summary, if one accepts the additional structural
assumptions on the driver, the exponential transformation approach
provides numerical approximation schemes for qgBSDE under weaker
smoothness conditions for the driver.

In this paper we aim to give a survey of these two approaches to
obtain numerical results for qgBSDE. Doing this, we always keep an eye
on one of the most important applications of qgBSDE, which consists
of providing a genuinely probabilistic approach to utility
optimization problems for exponential utility, or equivalently risk
minimization problems with respect to the entropic risk measure,
that lead to explicit descriptions of prices and hedges. We motivate
qgBSDE by reviewing a simple exponential utility optimization
problem resulting from a method to determine the utility
indifference price of an insurance related asset in a typical
incomplete market situation, following \cite{AIdR07} and
\cite{09Frei}. The setting of the problem allows in particular the
calculation of the driver of quadratic growth of the associated
BSDE. After discussing the problem of numerical approximations, in
this case by applying the method related to the exponential
transform, we are finally able to illustrate our findings by giving
some numerical simulations obtained with the resulting scheme.

The paper is organized as follows. In Section
\ref{section-preliminaries} we fix the notation used for treating
problems about qgFBSDE and recall some basic results. Section
\ref{section-money-applic} is devoted to presenting utility
optimization problems used for pricing and hedging derivatives on
non-tradable underlyings using correlated assets in a utility
indifference approach. In section \ref{section-path-reg-theo} we
review smoothness results for the solutions processes of qgFBSDE,
and apply them to show $L^2$-regularity of the control component
of the solution process of a qgFBSDE. In Section
\ref{section-trunc-procedure} we discuss the truncation method for
the quadratic terms of the driver to derive a numerical
approximation scheme for qgFBSDE. Section \ref{section-exp-transf}
is reserved for a discussion of the applicability of the exponential
transform in the qgBSDE setting. In Section
\ref{numerics-of-pricing-problem} we return to the motivating
pricing and hedging problem and use it as a platform for
illustrating our results by numerical simulations.

\section{Preliminaries}\label{section-preliminaries}

Fix $T\in\IR_+=[0,\infty)$. We work on the canonical Wiener space
$(\Omega, \cF,  \IP)$ on which a $d$-dimensional Wiener process $W =
(W^1,\cdots, W^d)$ restricted to the time interval $[0,T]$ is
defined. We denote by $\cF=(\cF_t)_{t\in[0,T]}$ its natural
filtration enlarged in the usual way by the $\IP$-zero sets.

Let $p\geq 2, m, d\in \IN$, $\IQ$ be a probability measure on $(\Omega,
\cF)$. We use the symbol $\IE^\IQ$ for the expectation with respect
to $\IQ$, and omit the superscript for the canonical measure $\IP$.
To denote the stochastic integral process of an adapted process $Z$
with respect to the Wiener process on $[0,T]$, we write
$Z*W=\int_0^\cdot Z_s \udws$.

For vectors $x = (x^1,\cdots, x^m)$ in Euclidean space $\R^m$ we
denote $|x| = (\sum_{i=1}^m (x^i)^2)^{\frac{1}{2}}$. In our analysis
the following normed vector spaces will play a role. We denote by

\begin{itemize}
\item $L^p(\R^m; \IQ)$ the space of $\cF_T$-measurable random
variables $X:\Omega\mapsto\R^m$, normed by $\lVert
X\lVert_{L^p}=\IE^\IQ[ \, |X|^p]^{\frac{1}{p}}$; $L^\infty$ the space
of bounded random variables;

\item  $\cS^p(\R^m)$ the space of all measurable processes $(Y_t)_{t\in[0,T]}$ with values in
$\R^m$ normed by $\| Y \|_{\cS^p} = \IE[\left( \sup_{t \in [0,T]}
|Y_t| \right)^{p}]^{\frac{1}{p}}$; $\cS^\infty(\R^m)$ the space of bounded measurable processes;

\item $\cH^p(\R^m, \IQ)$ the space of all progressively measurable processes $(Z_t)_{t\in[0,T]}$
with values in $\R^m$ normed by $\|Z\|_{\cH^p} = \IE^\IQ[\left( \int_0^T |Z_s|^2 \ud s \right)^{p/2} ]^{\frac{1}{p}};$

\item $BMO(\cF,\IQ)$ or $BMO_2(\cF,\IQ)$ the space of square integrable $\cF$-martingales $\Phi$ with $\Phi_0=0$ and we set
\[\lVert \Phi \lVert_{BMO(\cF,\IQ)}^2= \sup_{\tau}\Big\|\, \IE^\IQ\big[ \langle \Phi\rangle_T -
\langle \Phi \rangle_\tau| \cF_\tau \big]\Big\|_{\infty}< \infty,\]
where the supremum is taken over all stopping times $\tau\in[0,T]$.
More details on this space can be found in Appendix 1. In case $\IQ$ resp. $\cF$ is clear from the
context, we may omit the arguments $\IQ$ or $\cF$ and simply write
$BMO(\IQ)$ resp. $BMO(\cF)$ etc;

\item $\ID^{k,p}(\IR^d)$ and $\IL_{k,d}(\IR^d)$ the spaces of
Malliavin differentiable random variables and processes, see Appendix 2. 
\end{itemize}
In case there is no ambiguity about $m$ or $\IQ$, we may omit the
reference to $\IR^m$ or $\IQ$ and simply write $\cS^\infty$ or
$\cH^p$ etc.

We investigate systems of forward diffusions coupled with backward
stochastic differential equations with quadratic growth in the
control variable (qgFBSDE for short), i.e. given $x\in\IR^m$,
$t\in[0,T]$, and four continuous measurable functions $b$, $\sigma$,
$g$ and $f$ we analyze systems of the form
\begin{align}
\label{the-sde}
X^x_t&=x+\int_0^t b(s,X^x_s)\uds+\int_0^t \sigma(s,X^x_s)\udws,\\
\label{the-fbsde} Y^x_t&=g(X^x_T)+\int_t^T
f(s,X^x_s,Y^x_s,Z^x_s)\uds-\int_t^T Z^x_s\udws.
\end{align}
In case there is no ambiguity about the initial state $x$ of the
forward system, we may and do suppress the superscript $x$ and just
write $X, Y, Z$ for the solution components. For the coefficients of
this system we make the following assumptions: \bit
\item[{\bf (H0)}] $\quad$ There exists a positive constant $K$ such that $b,\sigma_i:[0,T]\times\IR^m\to\IR^m, 1\le i\le d,$ are uniformly Lipschitz continuous with Lipschitz constant $K$, and $b(\cdot,0)$ and $\sigma_i(\cdot,0), 1\le i\le d,$ are bounded
by $K$.

There exists a constant $M\in\IR_+$ such that $g:\R^m\to \R$ is
absolutely bounded by $M$, $f:[0,T]\times\R^m\times\R\times\R^d\to \R$ is measurable and continuous in $(x,y,z)$ and for $(t,x)\in[0,T]\times\IR^m$, $y,y'\in\IR$ and $z,z'\in\IR^d$ we have
\begin{align*}
|f(t,x,y,z)|&\leq M(1+|y|+|z|^2),\\
|f(t,x,y,z)-f(t,x,y',z')|&\leq M\Big\{|y-y'|+(1+|z|+|z'|)|z-z'|\Big\}.
\end{align*}
\eit The theory of SDE is well established. Since we wish to focus on the backward equation component of our system we emphasize that the relevant results for SDE are summarized in Appendix 3.

\begin{theorem}[Properties of qgFBSDE]\label{theo:moment-estimates-special-class}
Under {(H0)}, the system (\ref{the-sde}), (\ref{the-fbsde}) has a
unique solution $(X,Y,Z)\in \cS^2 \times \cS^\infty \times \cH^2$.
The respective norms of $Y$ and $Z$ can be dominated from above by constants depending only on $T$ and $M$ as given by
assumption {(H0)}. Furthermore
\[
Z*W=\int_0^. Z_s\udws \in BMO(\IP) \textrm{ and hence for all }p\geq 2\textrm{ one has }Z\in\cH^p.
\]
%
\end{theorem}
It is possible to go beyond the bounded terminal condition
hypothesis by imposing the existence of all its exponential moments
instead. In this case $Z*W$ is no longer in $BMO$. As we shall
see in Section \ref{section-smooth-results}, the $BMO$ property of
$Z*W$ plays a crucial role in all of our smoothness results for systems
of FBSDE. It combines with the inverse H\"older inequality for the
exponentials generated by $BMO$ martingales to control moments of
functionals of the solutions of FBSDE. Smoothness of solutions is
instrumental for instance in estimates for numerical approximations
of solutions.

\section{Pricing and hedging with correlated assets}\label{section-money-applic}

The pivotal task of mathematical finance is to provide solid
foundations for the valuation of contingent claims. In recent years,
markets have displayed an increasing need for financial instruments pegged
to non-tradable underlyings such as temperature and energy indices
or toxic matter emission rates. In the same manner as liquidly
traded underlyings, securities on non-tradable underlyings are used
to measure, control and manage risks, as well as to speculate and
take advantage of market imperfections. Since non-tradability
produces residual risks which are innate and inaccessible to
hedging, institutional investors look for tradable assets which are
correlated to the non-tradable ones. In incomplete markets, one
established pricing paradigm is the utility maximization principle.
Upon choosing a risk preference, investors evaluate contingent
claims by replicating according to an investment strategy that
yields the most favorable utility value. Interplays and connections
between the pricing of contingent claims on non-tradable underlyings
and the theory of qgFBSDE were studied, among others, by
\cite{07AIR}, \cite{06Mor}, \cite{09IRR}, and recently by
\cite{09Frei}. Based on this setup, we consider the problem of
numerically evaluating contingent claims based on non-tradable
underlyings. This will be done by intervention of the exponential
transformation of qgBSDE, to be introduced in Section
\ref{section-exp-transf}. It allows to work under weaker assumptions
than the numerical schemes for qgFBSDE based on the results reviewed
in Section \ref{section-smooth-results}, and will allow some
illustrative numerical simulations in Section
\ref{numerics-of-pricing-problem}.

\medskip

The following toy market setup can be found in Section 4 of
\cite{09Frei}. Assume $d=2$, so that $W = (W^1, W^2)$ is our basic
two-dimensional Brownian motion. We use them to define a third
Brownian motion $W^3$ correlated to $W^1$ with respect to a
correlation coefficient $\rho\in[-1,1]$ according to
\begin{align*}
W^3_s &:= \int_0^s \rho \mathrm{d}W^1_u + \int_0^s \sqrt{1-\rho^2}
\mathrm{d}W^2_u, ~~ 0 \leq s \leq T.
\end{align*}
Contingent claims are assumed to be tied to a one-dimensional non-tradable index that is subject to
\begin{align}\label{eq:nonTrad}
\mathrm{d} R_t &= \mu(t,R_t) \udt + \sigma(t,R_t) \mathrm{d}W^1_t, ~ R_0 =
r_0 > 0,
\end{align}
where $\mu,\sigma:[0,T] \times \R \to \R$ are deterministic
measurable and uniformly Lipschitz continuous functions, uniformly
of (at most) linear growth in their state variable. The securities
market is governed by a risk free bank account yielding zero
interest and one correlated risky asset whose dynamics (with respect
to the zero interest bank account \emph{num\'eraire}) are governed by
\begin{align}\label{eq:risky}
\frac{\mathrm{d} S_s}{S_s} &= \alpha(s,R_s) \uds + \beta(s,R_s)
\mathrm{d}W^3_s, ~ S_0 = s_0 > 0.
\end{align}
In compliance with \cite{07AIR}, we assume that $\alpha, \beta:
[0,T] \times \R \to \R$ are bounded and measurable functions, and
furthermore $\beta^2(t,r) \geq \varepsilon >0$ holds uniformly for
some fixed $\varepsilon>0$. Next, we set
\begin{align*}
\theta(s,r) &:= \frac{\alpha(s,r)}{\beta(s,r)}, ~~ (s,r) \in [0,T]
\times \R,
\end{align*}
and note that the conditions on $\alpha$ and $\beta$ imply that
$\theta$ is uniformly bounded.

An admissible investment strategy is defined to be a real-valued,
measurable predictable process $\lambda$ such that $\int_0^T
\lambda^2_u \udu < \infty$ holds $\IP$-almost surely and such that
the family
\begin{align}\label{eq:unifInt}
\left\{ e^{-\eta \int_0^\tau \lambda_u \frac{d S_u}{S_u}} : \tau ~
\text{stopping time with values in} ~ [0,T] \right\}
\end{align}
is uniformly integrable. The set of all admissible investment
strategies is denoted by $\cA$. In the following, let $t\in [0,T]$
denote a fixed time. Then the set of all admissible investment
strategies living on the time interval $[t,T]$ is defined
analogously and we denote it by $\cA_t$. Let $v_t$ denote the
investor's initial endowment at time $t$, that is, $v_t$ is an
$\cF_t$-measurable bounded random variable. The gain of the investor
at time $s \in [t,T]$, denoted by $G_s$, is subject to trading
according to investing $\lambda$ into the risky asset, and therefore
given by
\begin{align*}
\mathrm{d} G^\lambda_s &= \lambda_s \frac{dS_s}{S_s}, ~ G_t = 0.
\end{align*}
We focus on European style contingent claims, i.e. payoff profiles
resuming the form $F(R_T)$ where we assume, in accordance with
\cite{07AIR}, that $F:\R \to \R$ is measurable and bounded.  Moreover
the investor's risk assessment presumes that her utility preference
is reflected by the exponential utility function, so given a nonzero
constant risk attitude parameter $\eta$, the investor's utility
function is
\begin{align*}
U(x) &= - e^{-\eta x}, ~ x \in \R.
\end{align*}
The evolution of the investor's portfolio over the time interval
$[t,T]$ consists of her initial endowment $v_t$, her gains (or
losses) via her investment into the risky asset under an investment
strategy $\lambda$ and holding one share of the contingent claim
$F(R_T)$. Her objective is to find an investment strategy such that
her time-$t$ utility is maximized, i.e.\ her maximization problem is
given by
\begin{align}\label{eq:util}
V^F_t(v_t) &:= \sup \left\{ \IE \left[ U( v_t + G^\lambda_T + F(R_T)) \big| \cF_t \right] : \lambda \in \cA_t \right\} \nonumber\\
&= \exp \left\{ -\eta v_t  \right\} \sup \left\{ \IE \left[ U(
G^\lambda_T + F(R_T)) \big| \cF_t \right] : \lambda \in \cA_t
\right\}
\end{align}
For the sake of notational convenience, we write
\begin{align}\label{eq:lazy}
V_t^F := V_t^F (0) = \sup \left\{ \IE \left[ U( G^\lambda_T + F(R_T))
\big| \cF_t \right] : \lambda \in \cA_t \right\}.
\end{align}
Now pricing $F(R_T)$ within the utility maximization paradigm is
based on the identity
\begin{align*}
V^0_t(v_t) = V^F_t(v_t - p_t),
\end{align*}
where $V^0_t(v_t)$ denotes the time-$t$ utility with initial
endowment $v_t$ and with $F=0$ (see also Section 2 of \cite{07AIR}
and Section 3 of \cite{09Frei}). According to this identity, the
investor is indifferent about a portfolio with initial endowment
$v_t$ without receiving one quantity of the contingent claim
$F(R_T)$ and a portfolio with initial endowment $v_t - p_t$, now
receiving one quantity of the contingent claim in addition. Hence
$p_t$ is interpreted as the time-$t$ indifference price of the
contingent claim $F(R_T)$. By the equality $V^F_t(v_t) = \exp
\left\{ -\eta v_t \right\} V_t^F$, it follows that
\begin{align}\label{eq:indiffPrice}
p_t = \frac{1}{\eta} \log\frac{V^0_t}{V_t^F},
\end{align}
which means that the indifference price does not depend on the
initial endowment $v_t$. Since the time-$t$ indifference price
\eqref{eq:indiffPrice} is fully characterized by $V^0_t$ and
$V^F_t$, the focus now lies in the investigation of \eqref{eq:util}.
In fact, \cite{07AIR} and \cite{09Frei} have already pointed out
that \eqref{eq:lazy} yields a characterization by means of a
qgFBSDE. In accordance with \cite{09Frei}, let us denote by
$(\cG_u)_{0\leq u \leq T}$ the filtration generated by $W^1$,
completed by $\IP$-null sets. \cite{09Frei}'s main ideas for
rephrasing \eqref{eq:util} in terms of a qgBSDE are summarized in
the following
\begin{lemma}\label{lemma:Frei}
The qgFBSDE
\begin{alignat}{2}
& Y_s    & = & ~~ F(R_T) + \int_s^T f(u,R_u,Z_u) \udu - \int_s^T Z_u \mathrm{d}W^1_u,\quad s\in[0,T], \label{eq:utilBSDE}\\
& f(u,r,z) ~ & = & ~~ \frac{\theta^2(u,r)}{2\eta} - z\rho \theta(u,
r) - \frac{\eta}{2} \left( 1-\rho^2 \right) z^2,
\label{eq:utilDriver}
\end{alignat}
has a unique solution $(Y,Z)\in \cS^\infty \times \cH^{2}$ such that
$V^F_t = -e^{-\eta Y_t}$ holds $\IP$-almost surely.
\end{lemma}
\begin{proof}
Since $\theta(\cdot,r)$ is uniformly bounded and $\cG$-predictable,
the driver of \eqref{eq:utilBSDE} satisfies the conditions of
\cite{00Kob}; thus \eqref{eq:utilBSDE} admits a unique solution
$(Y,Z) \in \cS^\infty \times \cH^{2}$. Moreover, \cite{05ManSchweiz}
have shown that $Z * W^1$ is both a BMO($\cF$)- and
BMO($\cG$)-martingale. See also \cite{07AIR}. To prove the identity
$V^F_t = -e^{-\eta Y_t}$, we notice that
\begin{align*}
e^{ -\eta\left(G^\lambda_T + Y_T\right) } &= e^{ -\eta G^\lambda_t}
e^{ -\eta Y_t } e^{-\eta\left(Y_T - Y_t\right)}  e^{
-\eta\left(G^\lambda_T - G^\lambda_t\right)}\\
&= e^{ -\eta Y_t } e^{-\eta\left(Y_T - Y_t\right)}  e^{
-\eta\left(G^\lambda_T - G^\lambda_t\right)},
\end{align*}
because $G^\lambda_t=0$. We then have
\begin{align*}
&\exp\left\{ -\eta\left(Y_T - Y_t\right) \right\} \exp\{ -\eta(G^\lambda_T - G^\lambda_t) \}\\
&\qquad= \exp \left\{ -\eta \left( \int_t^T Z_u \mathrm{d}W^1_u + \int_t^T \lambda_u \beta(u,R_u)\mathrm{d}W^3_u 
+ \int_t^T \left[ \lambda_u \alpha(u,R_u) - f(u,R_u,Z_u) \right]
\udu \right) \right\}.
\end{align*}
Denoting $\cE_t^s (M) = \cE(M)_s / \cE(M)_t$ for $ t \leq s \leq T$
where $\cE(M)_s$ is the stochastic exponential of a given
semi-martingale $M$, we introduce
\begin{align*}
K_u := \frac{1}{2} \Big( \eta \left( \rho Z_u + \beta(u,R_u)
\lambda_u  \right) - \theta_u \Big)^2, ~~ t \leq u \leq T.
\end{align*}
Then a simple calculation yields
\begin{align*}
&\exp\Big\{-\eta\big(Y_T - Y_t\big)\Big\} \exp\Big\{ -\eta\big(G^\lambda_T - G^\lambda_t\big) \Big\}\\
&\qquad\qquad = \cE_t^T \left( \int -\eta Z \mathrm{d}W^1 - \int
\eta \lambda \beta(\cdot,R) \mathrm{d}W^3 \right)
\exp\left\{ \int_t^T K_u \udu \right\}.
\end{align*}
Since $\lambda \beta(\cdot,R)*W^3$ is a BMO-martingale, we can
condition with respect to the $\sigma$-algebra $\cF_t$ and get
\begin{align}\label{eq:ineq1}
\IE \left[ e^{-\eta \left( G^\lambda_T + F(R_T) \right) } ~ \big| ~ \cF_t \right] &= e^{-\eta Y_t} e^{\int_t^T K_u \udu} 
\geq e^{-\eta Y_t}.
\end{align}
By \eqref{eq:unifInt} and a localization argument, this inequality
holds for every $\lambda \in \cA_t$, and therefore we have $V^F_t
\leq - e^{-\eta Y_t}$. To prove equality, note that the inequality
\eqref{eq:ineq1} becomes an equality for $\tilde{\lambda}_u = -
\frac{\rho}{\beta(u,R_u)}Z_u + \frac{\theta(u,R_u)}{\eta
\beta(u,R_u)}$; this in conjunction with the observation that
\begin{align*}
\exp\left\{-\eta \tilde{\lambda}_u \frac{dS_u}{S_u}  \right\} &=
\exp\left\{-\eta G^{\tilde{\lambda}}_T  \right\}
= \exp\left\{-\eta  \left( G^{\tilde{\lambda}}_T - G^{\tilde{\lambda}}_t \right) \right\}\\
&= \cE_t^T \left( \int -\eta Z \mathrm{d}W^1 - \int \eta
\tilde{\lambda} \beta(\cdot,R) \mathrm{d}W^3 \right) \times
\exp\left\{-\eta\left(Y_T - Y_t\right)\right\}
\end{align*}
is the product of a bounded process and true $\cF$-martingale yields
that condition \eqref{eq:unifInt} is satisfied. Hence
$\tilde{\lambda} \in \cA_t$ and we have shown $V^F_t = -e^{-\eta
Y_t}$.
\end{proof}
The proof of the previous Lemma \ref{lemma:Frei} yields the
following
\begin{corollary}\label{coro:optStrategy}
The investment strategy
\begin{align}\label{eq:optStrategy}
\tilde{\lambda}_s := - \frac{\rho}{\beta(s,R_s)}Z_s +
\frac{\theta(s,R_s)}{\eta \beta(s,R_s)}, ~~ t \leq s \leq T,
\end{align}
where $Z$ is the control component of the solution to
\eqref{eq:utilBSDE}, belongs to $\cA_t$ and satisfies
\begin{align*}
\IE \left[ U( v_t + G^{\tilde{\lambda}_T} + F(R_T) \big| \cF_t
\right] &= \sup \left\{ \IE \left[ U( v_t + G^\lambda_T + F(R_T)
\big| \cF_t \right] : \lambda \in \cA_t \right\}= V^F_t(v_t).
\end{align*}
\end{corollary}
One application is given in the following example.
\begin{example}\label{example:PutOption}[Put option on kerosene, compare with Example 1.2 from \cite{07AIR}]
Facing recent considerable declines in world oil prices, companies
producing kerosene wish to partially cover their risk of such a
depreciation. European put options are an established financial
instrument to comply with this demand of risk covering. Since
kerosene is not traded in a liquid market,
derivative contracts on this underlying must be arranged on an
over-the-counter basis. Knowing that the price of heating oil is
highly correlated with the price of kerosene, the pricing and hedging
of a European put option on kerosene can be done by a dynamic
investment in (the liquid market of) heating oil. A numerical
treatment of this pricing problem will be displayed in Section
\ref{numerics-of-pricing-problem}.
\end{example}

\section{Smoothness and path regularity results}\label{section-smooth-results}

The principal aim of this paper is to survey some recent results on
the numerical approximation of prices and hedging strategies of
financial derivatives such as the liability $F(R_T)$ in the setting
of the previous section. As we saw, this leads us directly to
qgFBSDE. In the subsequent sections we shall discuss an approach
based on a truncation of the driver's quadratic part in the control
variable. It will be crucial to give an estimate for the error
committed by truncating. Our error estimate will be based on
smoothness results for the control component $Z^x$ of solutions of
the BSDE part of our system. Smoothness is understood both in the
sense of regular sensitivity to initial states $x$ of the forward
component, as well as in the sense of the stochastic calculus of
variations. Since the control component of the solution of a BSDE is
related to the Malliavin trace of the other component, we will be
led to look at variational derivatives of the first order.

Our first result concerns the smoothness of the map
$[0,T]\times\IR^m\ni (t,x)\mapsto Z_t^x$, especially its
differentiability in $x$. The second result refers to the
variational differentiability of $(Y^x, Z^x)$ in the sense of
Malliavin's calculus. We shall work under the following hypothesis,
where we denote the gradient by the common symbol $\nabla$, and by
$\nabla_u$ if we wish to emphasize the variable $u$ with respect to
which the derivative is taken.

\bit
\item[{\bf (H1)}] $\quad$ Assume that {(H0)} holds. For any $0\le t\le T$ the functions $b(t,\cdot),\sigma_i(t,
\cdot), 1\le i\le d,$ are continuously differentiable with bounded
derivatives in the spatial variable. There exists a positive
constant $c$ such that
\begin{align}\label{uniform-ellipticity}
y^T \sigma(t,x)\sigma^T(t,x)y \geq c|y|^2,\quad x,y\in\IR^m,\ t\in[0,T].
\end{align}

$f$ is continuously partially differentiable in $(x,y,z)$ and there exists
$M\in\IR_+$ such that for
$(t,x,y,z)\in[0,T]\times\IR^m\times\IR\times\IR^d$
\begin{align*}
|\nabla_x f(t,x,y,z)|&\leq M(1+|y|+|z|^2),\\
|\nabla_y f(t,x,y,z)|&\leq M,\\
|\nabla_z f(t,x,y,z)|&\leq M(1+|z|).
\end{align*}
$g:\R^m\to\R$ is a continuously differentiable function satisfying $|\nabla g|\leq M$.
\eit

\subsection*{Smoothness results}
The following differentiability results are extensions of Theorems
proved in \cite{AIdR07} and \cite{07BriCon}. For further details,
comments and complete proofs we refer to the mentioned works or to
\cite{phd-dosReis}.

\begin{theorem}[Classical differentiability]\label{theo:1st-order-diff-qgbsde}
Suppose that {(H1)} holds. Then for all $p\geq 2$ the solution
process $\Theta^x=(X^x,Y^x,Z^x)$ of the qgFBSDE (\ref{the-sde}),
(\ref{the-fbsde}) with initial vector $x\in\R^m$ for the forward
component belongs to $\cS^p\times \cS^p\times \cH^p$. The
application $\IR^m\ni x\mapsto (X^x,Y^x,Z^x)\in \cS^p(\IR^m)\times
\cS^p(\IR)\times \cH^p(\IR^d)$ is differentiable. The derivatives of
$x\mapsto X^x$ satisfy (\ref{diff-sde}) while the derivatives of the
map $x\mapsto (Y^x,Z^x)$ satisfy the linear BSDE
\begin{align*}
\nabla Y_t^{x}&= \nabla g(X_T^{x})\nabla X^x_T-\int_t^T \nabla Z_s^{x}\ud W_s+\int_t^T
\langle \nabla f(s,\Theta^x_s) , \nabla \Theta^x_s \rangle \ud s, \quad t\in[0,T].
\end{align*}
\end{theorem}

\begin{theorem}[Malliavin differentiability]\label{theo:bsde-1d-mall-diff}
Suppose that {(H1)} holds. Then the solution process $(X,Y,Z)$ of
FBSDE (\ref{the-sde}), (\ref{the-fbsde}) has the following
properties. For $x\in\IR^m$,
\begin{itemize}
\item
$X^x$ satisfies (\ref{mall-diff-of-X}) and for any $0\leq t\leq T$,
$x\in\R^m$ we have $(Y^x,Z^x)\in
\mathbb{L}_{1,2}\times\big(\mathbb{L}_{1,2}\big)^d$. $X^x$ fulfills
the statement of Theorem \ref{X-malliavin-diff}, and a version of
$(D_u Y^x_t,D_u Z^x_t)_{0\leq u,t\leq T}$ satisfies
\begin{align*}
D_u Y^x_t &= 0, \qquad D_u Z^x_t = 0,\qquad t<u\le T,\nonumber\\
D_u Y^x_t &= \nabla g(X^x_T) D_u X^x_T + \int_t^T \langle \nabla
f(s,\Theta^x_s), D_u \Theta^x_s \rangle \ud s - \int_t^T  D_u Z^x_s
\ud W_s, \quad t\in [u,T].
\end{align*}
Moreover, $(D_t Y^x_t)_{0\le t \leq T}$ defined by the above
equation is a version of $(Z^x_t)_{0\le t\leq T}$.
\item The following representation holds for any $0\leq  u\leq t\leq T$ and $x\in\R^m$
\begin{align*}
D_u Y^x_t &= \nabla_x Y^x_t (\nabla_x X^x_u)^{-1}\sigma(u,X^x_u),\quad a.s.,\\
Z_t&=\nabla_x Y^x_t (\nabla_x X^x_t)^{-1}\sigma(s,X^x_t),\quad a.s..
\end{align*}
\end{itemize}
\end{theorem}

\subsection*{Regularity and bounds for the solution process}\label{section-path-reg-theo}

A careful analysis of $DY$ in both its variables under the smoothness assumptions on the coefficients of our system
formulated earlier reveals the following continuity properties for the control process $Z$.
\begin{theorem}[Time continuity and bounds]\label{DY-continuity-Z-sinfty-bounds}
Assume {(H1)}. Then the control process $Z$ of the qgFBSDE (\ref{the-sde})-(\ref{the-fbsde}) has a continuous version on $[0,T]$. Furthermore for all $p\geq 2$ it satisfies
\begin{align}
\label{Z-moment-time} \lVert Z \lVert_{\cS^{p}}\ &<\infty.
\end{align}
\end{theorem}

\begin{theorem}[Regularity]\label{theo:path-reg-general-case}
Under {(H1)} the solution process $(X,Y,Z)$ of the qgFBSDE
(\ref{the-sde}), (\ref{the-fbsde}) satisfies for all $p \geq 2$ \bit
\item[i)]
there exists a constant $C_p>0$ such that for $0\leq
s\leq t\leq T$
 we have
\[\IE[\sup_{s\leq u\leq t} |Y_u-Y_s|^p \,]\leq C_p |t-s|^{\frac{p}{2}};\]

\item[ii)]
there exists a constant $C_p>0$ such that for any
partition $\pi = \{t_0,\cdots t_N\}$ with $0=t_0<\cdots<t_N=T$ of $[0,T]$ with mesh size $|\pi|$
\[
\sum_{i=0}^{N-1} \IE\Big[\Big(\int_{t_i}^{t_{i+1}}|Z_t-Z_{t_i}|^2\udt\Big)^{\frac{p}2}\Big]\leq C_p |\pi|^{\frac{p}2}.
\]
\eit
\end{theorem}

Now let $h=T/N$, $\pi^N=\{t_i=ih:i=0,\cdots,N\}$ be an equidistant
partition of $[0,T]$ with $N+1$ points and constant mesh size $h$.
Let $Z$ be the control component in the solution of the qgFBSDE
(\ref{the-sde}), (\ref{the-fbsde}) under {(H1)} and define the
family of random variables
\begin{align}\label{Z-bar-ti-pi}
 \bar{Z}^{\pi^N}_\ti&=\frac1{h}\IE\Big[\int_\ti^\tip
Z_s\uds\big|\cF_\ti\Big],\quad \ti\in\pi^N\setminus\{t_N\}.
\end{align}
For $0\le i\le N-1$ the random variable $\bar{Z}^{\pi^N}_\ti$ is the best $\cF_\ti$-measurable approximation of $Z$ in
$\cH^2([\ti,\tip])$, i.e.
\[
\IE\Big[ \int_\ti^\tip |Z_s-\bar{Z}^{\pi^N}_\ti|^2 \uds\Big]
=
\inf_{\Lambda} \IE\Big[\int_\ti^\tip |Z_s-\Lambda|^2 \uds\Big],
\]
where $\Lambda$ is allowed to vary in the space of all square integrable $\cF_\ti$-measurable random variables.
By constant interpolation we define $\bar{Z}^{\pi^N}_t = \bar{Z}^{\pi^N}_\ti$ for $t\in[\ti, \tip[$, $0\le
i\le N-1.$ It is easy to see that $(\bar Z^{\pi^N}_t)_{t\in[0,T]}$ converges to $(Z_t)_{t\in[0,T]}$ in
$\cH^2[0,T]$ as $h$ vanishes. Since $Z$ is adapted there exists a
family of adapted processes $Z^{\pi^N}$ indexed by our equidistant partitions
such that $Z^{\pi^N}_t=Z_\ti$ for $t\in [\ti,\tip)$ and that $Z^{\pi^N}$ converges to $Z$ in $\cH^2$ as $h$ tends to zero.
Since $\bar Z^{\pi^N}$ is the best
$\cH^2$-approximation of $Z$, we obtain
\[
\| Z-\bar Z^{\pi^N}  \|_{\cH^2}\leq \|Z-Z^{\pi^N} \|_{\cH^2}\to 0,\quad \textrm{as }h\to 0.
\]
The following Corollary of Theorem \ref{theo:path-reg-general-case} extends Theorem 3.4.3 in \cite{phd-zhang}
(see Theorem \ref{zhangs-path-reg-theo}) to the setting of qgFBSDE.
\begin{corollary}[$L^2$-regularity of $Z$]\label{theo:path-reg}
Under {(H1)} and for the sequence of equidistant partitions $(\pi^N)_{N\in\IN}$ of $[0,T]$ with mesh size $h=\frac{T}{N}$, we have
\begin{align*}
\max_{0\leq i\leq N-1}\Big\{\sup_{ t\in [\ti,\tip)} \IE\Big[\,|Y_t-Y_\ti|^2\Big]\Big\}+
\sum_{i=0}^{N-1} \IE\Big[\int_\ti^\tip |Z_s-\bar{Z}^{\pi^N}_\ti|^2\uds \Big]\leq C h,
\end{align*}
where $C$ is a positive constant independent of $N$.
\end{corollary}
\begin{remark}
The above corollary still holds if (H1) is weakened.
More precisely, the corollary's statement remains valid if one replaces in (H1) the sentence
\[
\textrm{``}g:\R^m\to\R\textrm{ is a continuously differentiable function satisfying }|\nabla g|\leq M.\textrm{''}
\]
by
\[
\textrm{``}g:\R^m\to\R\textrm{ is uniformly Lipschitz continuous in all its variables}.\textrm{''}
\]
The proof requires a regularization argument.
\end{remark}

\section{A truncation procedure}\label{section-trunc-procedure}

To the best of our knowledge so far none of the usual discretization
schemes for FBSDE has been shown to converge in the case of systems
of FBSDE considered in this paper, the driver of which is of
quadratic growth in the control variable. The regularity results
derived in the preceding section have the potential to play a
crucial role in numerical approximation schemes for qgFBSDE. We
shall now give arguments to substantiate this claim. In fact, the
regularity of the control component of the solution processes of our
BSDE will lead to precise estimates for the error
committed in truncating the quadratic growth part of the driver. We
will next explain how this truncation is done in our setting.

We start by introducing a sequence of real valued functions
$(\tilde{h}_n)_{n\in\IN}$ that truncate the identity on the real
line. For $n\in\IN$ the map $\tilde{h}_n$ is continuously
differentiable and satisfies \bit
\item $\tilde{h}_n\to \mbox{id}$ locally uniformly, $|\tilde{h}_n|\leq |\mbox{id}|$ and $|\tilde{h}_n|\leq n+1$;
moreover
\begin{equation*}
\tilde{h}_n(x)= \left\{
\begin{array}{cl}
(n+1)&,x> n+2,\\
x&,|x|\leq n,\\
-(n+1)&,x<-(n+2);
\end{array}
\right.
\end{equation*}
\item the derivative of $\tilde{h}_n$ is absolutely bounded by $1$ and converges to $1$ locally uniformly.
\eit We remark that such a sequence of functions exists. The above
requirements are for instance consistent with
\begin{displaymath}
\tilde{h}_n(x)= \left\{
\begin{array}{cl}
\big(-n^2+2nx-x(x-4)\big)/4&,x\in[n,n+2],\\
\big(n^2+2nx+x(x+4)\big)/4&,x\in[-(n+2),-n].\\
\end{array}
\right.
\end{displaymath}
We then define $h_n: \IR^d\to \IR^d$ by $z\mapsto
h_n(z)=(\tilde{h}_n(z_1),\cdots, \tilde{h}_n(z_d))$, $n\in\IN$. The
sequence $(h_n)_{n\in\IN}$ is chosen to be continuously
differentiable because the properties stated in Theorem
\ref{DY-continuity-Z-sinfty-bounds} need to hold for the solution
processes of the family of FBSDE that the truncation sequence
generates by modifying the driver according to the following
definition.

Recalling the driver $f$ of BSDE (\ref{the-fbsde}), for $n\in\IN$ we
define $f_n(t, x, y, z) := f(t,x,y,h_n(z))$, $(t,x,y,z) \in
[0,T]\times\IR^m\times \IR\times \IR^d$. With this driver and
(\ref{the-sde}) we obtain a family of truncated BSDE by
\begin{equation} \label{tr:quad-bsde-truncated}
Y^n_t = g(X_T) +\int_t^T f_n\big(s,X_s,Y^n_s,Z^n_s\big)   \uds
-\int_t^T Z^n_s\udws, \quad t\in[0,T], n\in\IN.
\end{equation}

The following Theorem proves that the truncation error leads to a
polynomial deviation of the corresponding solution processes in
their natural norms, formulated for polynomial order 12.

\begin{theorem}\label{theo-trunc-conv-rate}
Assume that {{(H1)}} is satisfied. Fix $n\in\IN$ and let $X$ be the
solution of (\ref{the-sde}). Let $(Y,Z)$ and $(Y^n,Z^n)_{n\in\IN}$
be the solution pairs of (\ref{the-fbsde}) and
(\ref{tr:quad-bsde-truncated}) respectively. Then for all $p\geq 2$
there exists a positive constant $C_{p}$  such that for all
$n\in\IN$
\[
\IE\Big[\sup_{t\in[0,T]} |Y^n_t-Y_t|^{p}\Big]+\IE\Big[\Big(\int_0^T
|Z^n_s-Z_s|^2\ud s\Big)^{\frac{p}2}\Big]\leq C_{p}\, \frac1{n^{12}}.
\]
\end{theorem}
The proof of Theorem \ref{theo-trunc-conv-rate} roughly involves
estimating the probability that $Z^n$ exceeds the threshold $n$ as a
function of $n\in\IN$ through Markov's inequality. The application
of Markov's inequality is possible thanks to (\ref{Z-moment-time}).

\section{The exponential transformation method}\label{section-exp-transf}

In the preceding sections we exhibited the significance of path
regularity for the solution of systems of qgFBSDE, in particular the
control component, for their numerical approximation. In this
section we shall discuss an alternative route to path regularity of
solutions in a particular situation that allows for weaker
conditions than in the preceding sections. We will use the
exponential transform known in PDE theory as the Cole-Hopf
transformation. This mapping takes the exponential of the component
$Y$ of a solution pair as the new first component of a solution pair
of a modified BSDE. It makes a quadratic term in the control
variable of the form $z\mapsto \gamma |z|^2$ vanish in the driver of
the new system. The price one has to pay for this approach is a
possibly missing global Lipschitz condition in the variable $y$ for
the modified driver. It is therefore not clear if the
new BSDE is amenable to the usual numerical discretization
techniques. We give sufficient conditions for the transformed driver
to satisfy a global Lipschitz condition. In this simpler setting our
techniques allow an easier access to smoothness results for the
solutions of the transformed BSDE. The Cole-Hopf transformation
being one-to-one, it is clear that regularity results carry over to
the original qgFBSDE.

Under (H0), we consider the transformation $P=e^{\gamma Y}$ and $Q =
\gamma P Z$. It transforms our qgBSDE (\ref{the-fbsde}) with driver
$f$ into the new BSDE
\begin{equation}\label{exp:BSDE}
P_t = e^{\gamma g(X_T)}
       +\int_t^T \Big[\gamma P_s f\Big(s,X_s,\frac{\log P_s}\gamma, \frac{Q_s}{\gamma P_s}\Big)
                        - \frac{1}{2} \frac{|Q|^2_s}{P_s}\Big]\uds
       -\int_t^T Q_s \udws,\quad t\in[0,T].
\end{equation}
Combining (\ref{exp:BSDE}) with SDE (\ref{the-sde}), we see that for
any $p\geq 2$ a unique solution $(X,P,Q)\in \cS^p\times
\cS^\infty\times \cH^p$ of (\ref{the-sde}) and (\ref{exp:BSDE})
exists. The properties of this triple follow from the properties of
the solution $(X,Y,Z)$ of the original qgFBSDE (\ref{the-sde}) and
(\ref{the-fbsde}). For clarity, we remark that since $Y$ is bounded, $P$ is also bounded and bounded away from 0. The latter property
allows us to deduce from the BMO martingale property of $Z*W$ the BMO
martingale property of $Q*W$. For the rest of this section we denote
by $\cK$ a compact subset of $(\delta,+\infty)$ for some constant $\delta\in\IR_+$ in which $P$ takes its values.

The form of the driver in (\ref{exp:BSDE}) indicates that after
transforming drivers of the form of the following hypothesis, we
have good chances to deal with a Lipschitz continuous one.  \bit
\item[{\bf (H0*)}] $\qquad$ Assume that (H0) holds. For $\gamma\in\IR$ let
$f:[0,T]\times\IR^m\times\IR\times\IR^d\to \IR$ be of the form
\[
f(t,x,y,z)=l(t,x,y)+ a(t,z) + \frac\gamma2 |z|^2,
\]
where $l$ and $a$ are measurable, $l$ is uniformly Lipschitz
continuous in $x$ and $y$, $a$ is uniformly Lipschitz continuous and
homogeneous in $z$, i.e. for $c\in\IR, (s,z)\in[0,T]\times \IR^d$ we
have $a(s,cz) = c a(s,z)$; $l$ and $a$ continuous in $t$.
\eit Assumption (H0*) allows us to simplify the BSDE obtained from the
exponential transformation to
\begin{align}\label{transformed-bsde}
P_t = e^{\gamma g(X_T)}
      +\int_t^T F(s, X_s, P_s, Q_s) \uds
      -\int_t^T Q_s \udws,\quad t\in[0,T],
\end{align}
where the driver is defined by
\begin{align}\label{transformed-driver}
F:[0,T]\times \IR^m\times  \cK \times \IR^d&\to \IR,\nonumber \\
(s,x,p,q)&\mapsto\, \gamma p\, l\big(s,x,\frac{\log p}\gamma) +
\gamma p\, a\big(s, \frac{q}{\gamma p}\big).
\end{align}
Thanks to the homogeneity assumption on $a$ our driver
simplifies further. Indeed, we have for $(s,x,p,q)\in [0,T]\times
\IR^m\times \IR\times \IR^d$
\begin{equation}\label{exp:driver-simple}
F(s,x,p,q)= \gamma p\, l\big(s,x,\frac{\log p}\gamma)
 + a\big(s,q\big).
\end{equation}

The terminal condition of the transformed BSDE still keeps the
properties it had in the original setting. Indeed, boundedness of
$g$ is inherited by $\exp(\gamma g)$. Furthermore, if $g$ is
uniformly Lipschitz, then clearly by boundedness of $g$,
the function $e^{\gamma g}$ is uniformly Lipschitz as well.

Let us next discuss the properties of the driver
(\ref{transformed-driver}) in the transformed BSDE. We recall that
since $l$ and $a$ are Lipschitz continuous, there is a constant
$C>0$ such that for all $(s,x,p,q)\in[0,T]\times \IR^m\times \cK
\times \IR^d$ we have
\begin{align*}
|F(s,x,p,q)|&\leq \big|\gamma p\, l\big(s,x,\frac{\log p}\gamma\big)
+ a\big(s, q\big)\big|\\
&\leq C |p|\big(1+|x|+|\log p\,|+|q|\big)\leq C\big(1+|x|+|p|+|q|\big).
\end{align*} This means that $F$ is of linear growth in $x, p$ and
$q$.

To verify Lipschitz continuity properties of $F$ in its variables
$x, p$ and $q$, by (\ref{exp:driver-simple}) and the Lipschitz
continuity assumptions on $a$, it remains to verify that
\[(x,p) \mapsto \gamma p\, l(s,x,\frac{\log p}{\gamma})\] is
Lipschitz continuous in $x$ and $p$, with a Lipschitz constant
independent of $s\in[0,T].$ As for $x$, this is an immediate
consequence of the Lipschitz continuity of $l$ in $x$. For $p$
we have to recall that $p$ is restricted to a compact set $\cK
\subset \IR_+$ not containing 0, to be able to appeal to the
Lipschitz continuity of $l$ in $y$. This shows that $F$ is globally
Lipschitz continuous in its variables $x, p$ and $q$.

We may summarize these observations in the following Theorem.
\begin{theorem}\label{transform:Lipschitz}
Let $f:[0,T]\times\IR^m\times \IR\times \IR^d \to \IR$ be a
measurable function, continuous on $\IR^m\times\IR\times\IR^d$, and
satisfying (H0*). Then $F$ as defined by (\ref{transformed-driver})
is a uniformly Lipschitz continuous function in the spatial
variables.
%
\end{theorem}

%
%
%

Theorem \ref{transform:Lipschitz} opens another route to tackle
convergence of numerical schemes via path regularity of the control
component of a solution pair of a qgFBSDE system whose driver
satisfies (H0*). Look at the new BSDE after applying the Cole-Hopf
transform. Since it possesses a Lipschitz continuous driver, path
regularity for the control component $Q$ of the transformed BSDE
will follow from Zhang's path regularity result stated in
(\ref{zhangs-path-reg-theo}) provided the driver is
$\frac{1}{2}$-H\"older continuous in time. Of course, by the
smoothness of the Cole-Hopf transform, the control component $Z$ of
the original BSDE will inherit path regularity from $Q$. This way we
circumvent the more stringent assumption (H1) which was made in
section \ref{section-smooth-results}.

In what follows the triples $(X,Y,Z)$ and $(X,P,Q)$ will always
refer to the solution of qgFBSDE (\ref{the-sde}), (\ref{the-fbsde})
and FBSDE (\ref{the-sde}), (\ref{transformed-bsde}) respectively.
\begin{theorem}
Let (H0*) hold. Assume that \[[0,T]\times \IR^m\times \cK
\times\IR^d\ni (s,x,p,q)\mapsto F(s,x,p,q)\in \IR,\] the driver of
BSDE (\ref{transformed-bsde}), is uniformly Lipschitz in $x, p$ and
$q$ and is $\frac12$-H\"older continuous in $s$. Suppose further
that the map $g:\IR^d\to \IR$, as indicated in (H0), is globally
Lipschitz continuous with Lipschitz constant $K$. Let $(X,Y,Z)$ be
the solution of qgFBSDE (\ref{the-sde}), (\ref{the-fbsde}), and
$\varepsilon>0$ be given. There exists a positive constant $C$ such
that for any partition $\pi=\{t_0,\cdots,t_N\}$ with $0=t_0, T=t_N,
t_0<\cdots<t_N$ of the interval $[0,T]$, with mesh size $|\pi|$ we
have
\begin{align*}
&\max_{0\leq i\leq N-1}\Big\{\sup_{ t\in [\ti,\tip)} \IE\Big[\,|Y_t-Y_\ti|^2\Big]\Big\}\leq C|\pi|
\quad\textrm{and}\quad
\sum_{i=0}^{N-1} \IE\Big[\int_\ti^\tip |Z_s-\bar{Z}^\pi_\ti|^2\uds \Big]\leq C |\pi|^{1-\varepsilon}.
\end{align*}
Moreover, if the functions $b$ and $\sigma$ are continuously
differentiable in $x\in\IR^m$ then $t\mapsto Z_t$ is a.s. continuous
in $[0,T]$.
\end{theorem}
\begin{proof}
Throughout this proof $C$ will always denote a positive constant the
value of which may change from line to line. Let $(X,P,Q)$ be the
solution of (\ref{the-sde}) and (\ref{transformed-bsde}), where $P$
takes its values in $\cK$ and $Q*W$ is a BMO martingale. Applying
Theorem \ref{zhangs-path-reg-theo} yields a positive constant $C$
such that for any partition $\pi=\{t_0,\cdots,t_N\}$ of $[0,T]$ with
mesh size $|\pi|$
\begin{align*}
&\max_{0\leq i\leq N-1}\Big\{\sup_{ t\in [\ti,\tip)} \IE\Big[\,|P_t-P_\ti|^2\Big]\Big\}+
\sum_{i=0}^{N-1} \IE\Big[\int_\ti^\tip |Q_s-\bar{Q}^\pi_\ti|^2\uds \Big]\leq C |\pi|.
\end{align*}
Since $P$ takes its values in the compact set $\cK\subset \IR_+$ not
containing 0 there exists a constant $C$ such that for any $0\le
i\le N-1, t\in [\ti, \tip)$
\begin{align*}
|Y_t-Y_\ti|=C|\log P_t-\log P_\ti|\leq C |P_t - P_\ti|.
\end{align*}
Using the two above inequalities we have
\begin{align*}
&\max_{0\leq i\leq N-1}\Big\{\sup_{ t\in [\ti,\tip)} \IE\Big[\,|Y_t-Y_\ti|^2\Big]\Big\}\leq C \max_{0\leq i
\leq N-1}\Big\{\sup_{ t\in [\ti,\tip)} \IE\Big[\,|P_t-P_\ti|^2\Big]\Big\}
\leq C |\pi|.
\end{align*} This proves the first inequality. For the second one,
note that by definition for $0\le i\le N-1, t\in [\ti, \tip)$
\begin{align*}
|Z_t-\bar{Z}_\ti|\leq |Z_t-Z_\ti|
&\leq \frac1\gamma\Big \{|\frac{Q_t}{P_t}-\frac{Q_t}{P_\ti}|+|\frac{Q_t}{P_\ti}-\frac{Q_\ti}{P_\ti}|\Big\}
\leq \frac1\gamma \Big\{|Q_t||\frac{1}{P_t}-\frac{1}{P_\ti}|+\frac{1}{|P_\ti|}|Q_t-Q_\ti|\Big\}\\
&\leq C \Big\{ |Q_t|\ |P_t-P_\ti|+|Q_t-Q_\ti|\Big\}.
\end{align*}
We therefore have for $0\le i\le N-1$
\begin{align*}
& \IE\Big[\int_\ti^\tip |Z_s-\bar{Z}^\pi_\ti|^2\uds \Big]\leq  \IE\Big[\int_\ti^\tip |Z_s-Z_\ti|^2\uds \Big]\\
&\qquad \leq 2C\Big\{
\IE\Big[\sup_{t\in[\ti,\tip)}|P_t-P_\ti|^2\,\int_\ti^\tip |Q_s|^2\uds
\Big] +\IE\Big[\int_\ti^\tip|Q_t-Q_\ti|^2\uds \Big]\Big\}.
\end{align*}
Since  $Q\in\cH^p$ for all $p\geq 2$, for any two real numbers
$\alpha,\beta \in(1,\infty)$ satisfying $1/\alpha+1/\beta=1$ we may
continue using H\"older's inequality on the right hand side of the
inequality just obtained, and then Theorem
\ref{zhangs-path-reg-theo} to the term containing $P$. This yields
the following inequality valid for any $0\le i\le N-1$ with a
constant $C$ not depending on $i$
\begin{align*}
 \IE\Big[\int_\ti^\tip |Z_s-\bar{Z}^\pi_\ti|^2\uds \Big]
&\leq C\Big\{\IE\big[\sup_{t\in[\ti,\tip)}|P_t-P_\ti|^{2\alpha}\big]^{\frac1\alpha}
\IE\big[\Big(\int_\ti^\tip |Q_s|^2\uds\Big)^\beta \Big]^\frac1\beta
+|\pi|\Big\}\\
&\leq C\Big\{
\IE\big[\sup_{t\in[\ti,\tip)}|P_t-P_\ti|^{2}\big]^{\frac1\alpha} +
|\pi| \Big\}\leq C\Big\{ |\pi|^{\frac1\alpha} + |\pi|\Big\}.
\end{align*}
Now choose $\alpha = \frac{1}{1-\varepsilon},$ to complete the
claimed estimate.

To prove that $Z$ admits a.s. a continuous version, it is enough to
remark that the Theorem's assumptions imply the conditions of
Corollary 5.6 in \cite{MZ02-pathreg}. The referred result yields
that $Q$ is a.s. continuous on $[0,T]$. Since $P$ is continuous and
bounded away from zero we conclude from the equation $\gamma P Z =
Q$ that $Z$ is a.s. continuous as well.
\end{proof}

\section{Back to the pricing problem}\label{numerics-of-pricing-problem}
We now come back to the numerical valuation of the put option on
kerosene as depicted in example \ref{example:PutOption}. Notations
in the following are adopted from Section
\ref{section-money-applic}. Assume that the put option expires at
$T=1$. Let $R$ and $S$ denote the dynamics for the financial value
of kerosene and heating oil respectively. In particular we assume
both dynamics to be lognormally distributed according to
\begin{align*}
\mathrm{d} R_t
&=\mu(t,R_t) \udt + \sigma(t,R_t)\mathrm{d}W^1_t
 = 0.12\, R_t\,\udt + 0.41\, R_t \,\mathrm{d}W^1_t, \\
\frac{\mathrm{d}S_t}{S_t}
&= \alpha(t,R_t) \udt + \beta(t,R_t)\mathrm{d}W^3_t
 = 0.1 \, \udt + 0.35\, \mathrm{d}W^3_t,
\end{align*}
and we assume the spot price for heating oil to be $s_0 = 173$ money
units (e.g. US Dollar, Euro), see also equations \eqref{eq:nonTrad}
and \eqref{eq:risky}. Risk aversion is set at the level of
$\eta = 0.3$. Figure \ref{pic:TradVsNontrad} displays sample paths
of the kerosene price with a spot price of $r_0=170$ and heating oil price
at different correlation levels using the explicit solution formula
for the geometric Brownian motion. We see that the higher the
correlation, the better the approximation of the kerosene by heating
oil becomes.
\begin{figure}[!hbt]
\centering
\includegraphics[scale=0.5]{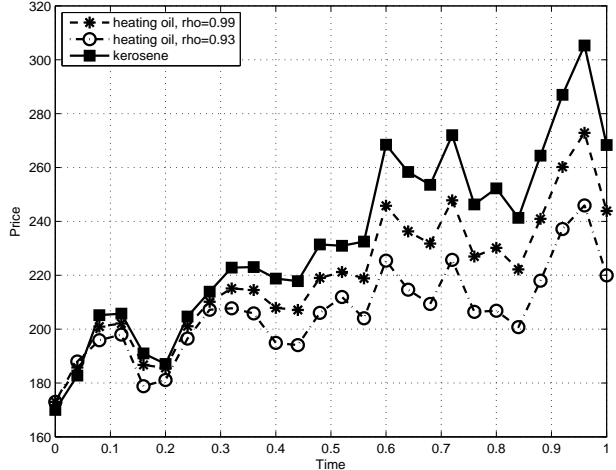}     
\caption{Price paths of the nontradable asset kerosene and the correlated asset heating oil at different correlation levels. The spot of kerosene was set to $r_0 = 170$.}\label{pic:TradVsNontrad}
\end{figure}
We have seen that the valuation of the put option via utility
maximization yields the pricing formula \eqref{eq:indiffPrice} which
in conjunction with Lemma \ref{lemma:Frei} becomes the difference of
two solutions of a qgBSDE with the generator  \eqref{eq:utilDriver}
\begin{align*}
p_t &= Y^F_t - Y^0_t, ~~ 0 \leq t \leq T,
\end{align*}
where $F(x) = (K-x)^+$ for some strike $K>0$. For the numerical
simulation of the qgFBSDE $Y^F$ and $Y^0$, we apply the exponential
transformation to both BSDE (see Section \ref{section-exp-transf})
and then employ the algorithm by \cite{07BD} with $N=100$
equidistant time points, $70000$ paths and a regression basis
consisting of five monomials and the payoff function of the put
option. The Picard iteration stops as soon as the difference of two
subsequent time zero values is less than $10^{-5}$. Simulations
reveal that $12$ to $13$ iterations are needed for solving one
exponentially transformed qgFBSDE.
\begin{figure}[!hbt]
\centering
\subfigure[Put option price in terms varying strikes at a fixed kerosene spot $r_0=170$.]
{
 \label{pic:PriceVsStrike}
  \includegraphics[width=7.3cm]{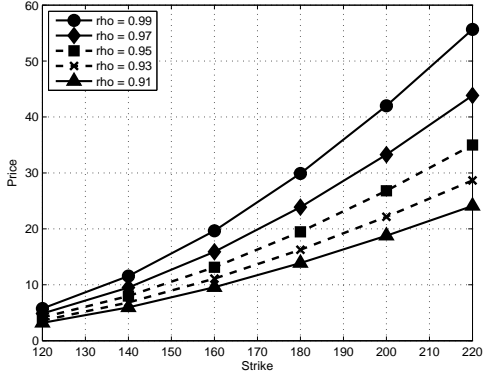}        
}
\hspace{0.2cm}
\subfigure[Put option price in terms of varying kerosene spots at a fixed strike $K=200$.] 
{
 \label{pic:PriceVsSpot}
  \includegraphics[width=7.3cm]{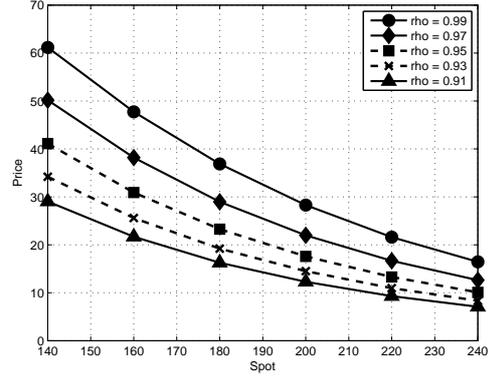}           
}
\caption{Values of the put option in terms of kerosene spot and strike for varying correlations. High correlations lead to high the prices for the contingent claim.}
\end{figure}
Figures \ref{pic:PriceVsStrike} and \ref{pic:PriceVsSpot} depict the
time zero price $p_0$ of the put option at different strike and
kerosene spot levels. The lower the correlation, the lower the price
becomes. This is clear because lower correlations between heating
oil and kerosene lead to higher non-hedgeable residual risk which
diminishes the risk covering effect of the contingent claim and thus
also its value.
\begin{figure}[!hbt]
\centering
\subfigure[Dynamics of the price process $p_t$ for strike $K = 180$.]
{
 \label{pic:PriceDyn}
  \includegraphics[width=7.3cm]{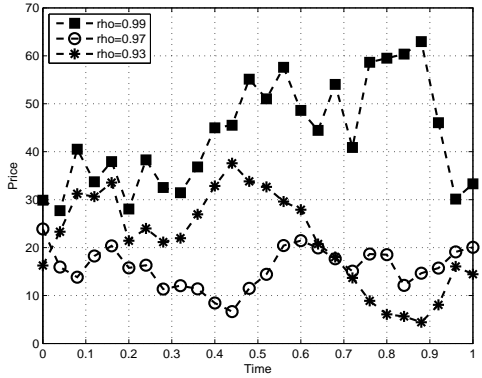}      
}
\hspace{0.5cm}
\subfigure[Dynamics of the optimal investment strategy $\pi_t$ for strike $K= 180$.] 
{
 \label{pic:OptStrategy}
  \includegraphics[width=7.3cm]{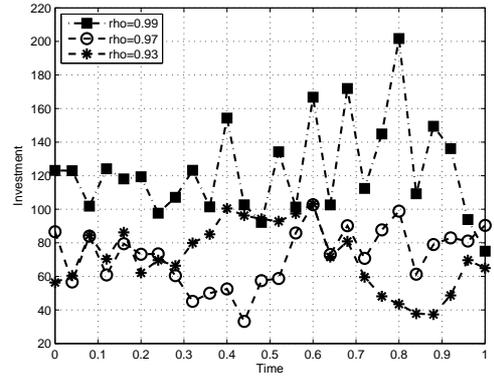}            
}
\caption{Paths of the price $p_t$ and the optimal investment strategy $\pi_t$ for varying correlation levels. In general high correlations entail greater market activity.}
\end{figure}
Figures \ref{pic:PriceDyn} and \ref{pic:OptStrategy} depict sample
paths of the dynamics for the price $p_t$ and the optimal investment
strategy $\pi_t$ for an at the money put with strike $K=180$ and
kerosene spot $r_0 = 170$. The plots depict price and monetary
investment for every fourth time point of the discretization. The
price process and the dynamics of the optimal investment strategy
are intertwined: high fluctuations of the price process result in
high fluctuations of the investment strategy and vice versa. In
general we observe that replication on high correlation levels tends to
entail greater market activity because kerosene price risks can
then be well hedged by market transactions that move closely along
the dynamics of heating oil. In contrast, replication on lower
correlation levels leads to a higher amount of residual risk which is inaccessible for hedging and thus lower market activity is needed.

\section*{Appendix 1 -- Some results on BMO martingales} \label{bmo-subsection}
\addcontentsline{toc}{section}{Appendix 1 -- Some results on BMO martingales}
BMO martingales play a key role for a priori estimates needed in our
sensitivity analysis of solutions of BSDE. For details about this
theory we refer the reader to \cite{Kazamaki1994}.

Let $\Phi$ be a $BMO(\cF, \IQ)$ martingale with $\Phi_0=0$. $\Phi$
being square integrable, the martingale representation Theorem
yields a square integrable process $\phi$ such that $\Phi_t=\int_0^t
\phi_s \udws, t\in[0,T]$. Hence the $BMO(\cF, \IQ)$ norm of $\Phi$
can be alternatively expressed as
\[\sup_{\tau\,\cF-\mbox{\scriptsize stopping time in}\,[0,T]} \IE^\IQ\Big[ \int_\tau^T \phi_s^2 \uds|\cF_\tau \Big]< \infty.\]
\begin{lemma}[Properties of BMO martingales]\label{bmoeigen}
Let $\Phi$ be a $BMO$ martingale. Then we have:
\begin{itemize}
  \item[1)] The stochastic exponential $\cE(\Phi)$ is uniformly
  integrable.
  \item[2)] There exists a number $r>1$ such that $\cE(\Phi_T)\in L^r$. This property follows from the \emph{Reverse H\"older inequality}.
  The maximal $r$ with this property can be expressed explicitly in terms of the BMO norm of $\Phi$.
  \item[3)] If
  $\Phi=\int_0^\cdot \phi_s \uds$
  has BMO norm $C$, then for $p\geq 1$ the following estimate holds
\begin{align*}
\IE[\Big(\int_0^T |\phi_s|^2\uds\Big)^p ]\leq 2p!(4C^2)^p.
\end{align*} Hence $BMO \subset \cH^p$ for all $p\geq 1$.
\end{itemize}
\end{lemma}

\section*{Appendix 2 -- Basics of Malliavin's calculus} \label{malliavin-calculus}
\addcontentsline{toc}{section}{Appendix 2 -- Basics of Malliavin's calculus}

We briefly introduce the main notation of the stochastic calculus of
variations also known as Malliavin's calculus. For more details, we
refer the reader to \cite{nualart2006}. Let ${\bf \cS}$ be the space
of random variables of the form
\[
\xi = F\Big((\int_0^T h^{1,i}_s \ud W^1_s)_{1\le i\le
n},\cdots,(\int_0^T h^{d,i}_s \ud W^d_s)_{1\le i\le n})\Big),
\] where $F\in C_b^\infty(\R^{n\times d})$, $h^1,\cdots,h^n\in L^2([0,T]; \R^d)$, $n\in\IN.$
To simplify notation, assume that all $h^j$ are written as row
vectors.
For $\xi\in {\bf \cS}$, we define $D = (D^1,\cdots, D^d):{\bf \cS}\to L^2(\Omega\times[0,T])^d$ by
\[
D^i_\theta \xi = \sum_{j=1}^n \frac{\partial F}{\partial x_{i,j}}
\Big( \int_0^T h^1_t \ud W_t,\ldots,\int_0^T
h^n_t\ud W_t\Big)h^{i,j}_\theta,\quad 0\leq \theta\leq T,\quad 1\le
i\le d,\]
and for $k\in\IN$ its $k$-fold iteration by
\[D^{(k)} = (D^{i_1}\cdots D^{i_k})_{1\le i_1,\cdots, i_k\le d}.\]
For $k\in\IN, p\ge 1$ let $\ID^{k,p}$ be the closure of $\cS$ with respect to the norm
\[
\lVert \xi \lVert_{k,p}^p = \IE\Big[\|\xi\|^p_{L^p}+ \sum_{i=1}^{k}\|| D^{(k)]} \xi| \|_{(\cH^p)^i}^p\Big].
\]
$D^{(k)}$ is a closed linear operator on the space $\mathbb{D}^{k,p}$. Observe
that if $\xi\in \ID^{1,2}$ is $\cF_t$-measurable then $D_\theta \xi=0$ for $\theta
\in (t,T]$. Further denote $\mathbb{D}^{k,\infty}=\cap_{p>1}\mathbb{D}^{k,p}$.

We also need Malliavin's calculus for smooth stochastic processes with values in $\R^m.$ For $k\in\IN, p\ge 1,$ denote by $\mathbb{L}_{k,p}(\R^m)$ the
set of $\R^m$-valued progressively measurable processes
$u = (u^1,\cdots, u^m)$ on $[0,T]\times \Omega$ such that
\begin{itemize}
\item[i)]
For Lebesgue-a.a. $t\in[0,T]$, $u(t,\cdot)\in(\mathbb{D}^{k,p})^m$;
\item[ii)] $[0,T]\times \Omega \ni (t,\omega)\mapsto D^{(k)} u(t,\omega)\in (L^2([0,T]^{1+k}))^{d\times n}$ admits a progressively measurable
version;
\item[iii)] $\lVert u \lVert_{k,p}^p=  \|u\|_{\cH^p}^p+\sum_{i=1}^{k} \| \,D^i u\,\|_{(\cH^p)^{1+i}}^p\,<\infty$.
\end{itemize}
%
Note that Jensen's inequality gives for all $p\geq 2$
\[
\IE\Big[\Big( \int_0^T \int_0^T|D_u X_t|^2\ud u\,\ud t\Big)^{\frac{p}{2}} \Big]
\leq
T^{p/2-1}\int_0^T  \| D_u X\|_{\cH^p}^p
\ud u.
\]


\section*{Appendix 3 -- Some results on SDE}\label{results-on-sdes}
\addcontentsline{toc}{section}{Appendix 3 -- Some results on SDE}
We recall results on SDE known from the literature that are relevant
for this work. We state our assumptions in the multidimensional
setting. However, for ease of notation we present some formulas in
the one dimensional case.
\begin{theorem}[Moment estimates for SDE]
Assume that {(H0)} holds. Then (\ref{the-sde}) has a unique solution
$X\in\cS^2$ and the following moment estimates hold: for any $p\geq
2$ there exists a constant $C>0$, depending only on $T$, $K$ and $p$
such that for any $x\in\R^m, s,t\in[0,T]$
\begin{align*}
\IE[\,\sup_{0\leq t\leq T} |X_t|^p\,]&\leq C\IE\Big[\, |x|^p+\int_0^T \big(|b(t,0)|^p+|\sigma(t,0)|^p \big)\ud t \Big],\\
\IE[\sup_{s\leq u\leq t}|X_u-X_s|^p\,]&\leq C \IE\Big[\,|x|^p+\sup_{0\leq t\leq T}\big\{|b(t,0)|^p+|\sigma(t,0)|^p\big\}  \Big]\,|t-s|^{p/2}.
\end{align*}
Furthermore, given two different initial conditions $x,x'\in\IR^m$,
we have
\begin{align*}
\IE\Big[\sup_{0\leq t\leq T}| X_t^x-X_t^{x'} |^p\Big]\leq C |x-x'|^p.
\end{align*}
\end{theorem}

\begin{theorem}[Classical differentiability]\label{theo.SDE-classic-diff}
Assume {(H1)} holds. Then the solution process $X$ of
(\ref{the-sde}) as a function of the initial condition $x\in\IR^m$
is differentiable and satisfies for $t\in[0,T]$
\begin{align}\label{diff-sde}
\nabla X_t &= I_m +\int_0^t \nabla b(X_s) \nabla X_s \uds +\int_0^t
\nabla \sigma(X_s)\nabla X_s\udws,
\end{align}
where $I_m$ denotes the $m\times m$ unit matrix. Moreover, $\nabla
X_t$ as an $m\times m$-matrix is invertible for any $t\in[0,T]$. Its
inverse $(\nabla X_t)^{-1}$ satisfies an SDE and for any $p\geq 2$
there are positive constants $C_p$ and $c_p$ such that
\begin{align*}
\| \nabla X \|_{\cS^p}+\| (\nabla X)^{-1}\|_{\cS^p}&\leq C_p
\end{align*}
and
\begin{align*}
\IE\Big[\sup_{s\leq u\leq t}|(\nabla X_u)-(\nabla X_s)|^p+\sup_{s\leq u\leq t}|(\nabla X_u)^{-1}-(\nabla X_s)^{-1}|^p\,\Big]&\leq c_p \,|t-s|^{p/2}.
\end{align*}
\end{theorem}

\begin{theorem}[Malliavin Differentiability]\label{X-malliavin-diff}
Under {(H1)}, $X\in \mathbb{L}_{1,2}$ and its Malliavin derivative
admits a version $(u,t)\mapsto D_u X_t$ satisfying for $0\le u\le
t\le T$ the SDE
\begin{align}\label{mall-diff-of-X}
D_u X_t &= \sigma(X_{u})+\int_u^t \nabla b(X_s) D_u X_s \uds
+\int_u^t \nabla \sigma(X_s) D_u X_s\udws.
\end{align}
Moreover, for any $p\geq 2$ there is a constant $C_p>0$ such that for $x\in\R^m$ and $0\leq v \leq u\leq t\leq s\leq T$
\begin{align*}
\| D_u X\|_{\cS^p}^p&\leq C_p(1+|x|^p),\\
\IE[\, |D_u X_t - D_u X_s|^p ]&\leq  C_p(1+|x|^p)|t-s|^{\frac{p}{2}},\\
\| D_u X-D_v X\|_{\cS^p}^p&\leq  C_p(1+|x|^p)|u-v|^{\frac{p}{2}}.
\end{align*}
By Theorem \ref{theo.SDE-classic-diff}, we have the representation
\begin{align*}
D_u X_t= \nabla X_t (\nabla X_u)^{-1}\sigma(X_u)\1_{[0,u]}(t),\quad \textrm{ for all }u,t\in[0,T].
\end{align*}
\end{theorem}

\section*{Appendix 4 -- Path regularity for Lipschitz FBSDE}\label{other-results-on-FBSDE}
\addcontentsline{toc}{section}{Appendix 4 -- Path regularity for Lipschitz FBSDE}
We state a version of the $L^2$-regularity result for FBSDE
satisfying a global Lipschitz condition. The result which was seen
to be closely related to the convergence of numerical schemes for
systems of FBSDE is due to \cite{phd-zhang}. For our FBSDE system
(\ref{the-sde}), (\ref{the-fbsde}) we assume that $b,\sigma,f,g$ are deterministic measurable functions that are Lipschitz continuous with respect to the spatial variables and $\frac12$-H\"older continuous with respect to time. Furthermore we assume that $\sigma$ satisfies
(\ref{uniform-ellipticity}). Then from \cite{97KPQ} one easily
obtains existence and uniqueness of a solution triple $(X,Y,Z)$ of
FBSDE (\ref{the-sde}), (\ref{the-fbsde}) belonging to $\cS^2\times
\cS^2\times \cH^2$. For a partition $\pi$ of $[0,T]$ define the
process $\bar{Z}^\pi$ as in (\ref{Z-bar-ti-pi}). Then the following
result holds.
\begin{theorem}[Path regularity result of \cite{phd-zhang}]\label{zhangs-path-reg-theo}
Let $(X,Y,Z)\in\cS^2\times \cS^2\times\cH^2$ be the solution of FBSDE (\ref{the-sde}), (\ref{the-fbsde}) in the setting described above. Then there exists $C\in\IR_+$ such that for any partition
$\pi = \{t_0,\cdots,t_N\}$ of the time interval $[0,T]$ with mesh
size $|\pi|$ we have
\begin{align*}
&\max_{0\leq i\leq N-1}\Big\{\sup_{ t\in [\ti,\tip)} \IE\Big[\,|Y_t-Y_\ti|^2\Big]\Big\}\\
&\qquad\quad+\sum_{i=0}^{N-1} \IE\Big[\int_\ti^\tip |Z_s-\bar{Z}^\pi_\ti|^2\uds \Big]
+\sum_{i=0}^{N-1} \IE\Big[\int_\ti^\tip |Z_s-Z_\ti|^2\uds \Big]
\leq C |\pi|.
\end{align*}
\end{theorem}

\bigskip

\section*{Acknowledgements}
Gon\c calo Dos Reis would like to thank both Romuald Elie and Emmanuel Gobet for the helpful discussions. Jianing Zhang acknowledges financial support by IRTG 1339 SMCP.

\bibliographystyle{abbrvnat}                        

\end{document}